\documentclass{article}

\usepackage{microtype}
\usepackage{graphicx}
\usepackage{subcaption}
\usepackage{booktabs} 

\usepackage{hyperref}
\usepackage{xcolor}
\usepackage{minted}

\usepackage[preprint]{icml2026}

\usepackage{amsmath}
\usepackage{amssymb}
\usepackage{mathtools}
\usepackage{amsthm}

\theoremstyle{plain}
\newtheorem{theorem}{Theorem}[section]

\theoremstyle{definition}
\newtheorem{definition}{Definition}

\theoremstyle{remark}

\usepackage[textsize=tiny]{todonotes}

\usepackage{mathrsfs}
\usepackage{bbm}
\usepackage{bm}

\usepackage{amsmath,amsfonts,bm}

\def\eqref#1{equation~\ref{#1}}

\def\1{\bm{1}}

\DeclareMathAlphabet{\mathsfit}{\encodingdefault}{\sfdefault}{m}{sl}
\SetMathAlphabet{\mathsfit}{bold}{\encodingdefault}{\sfdefault}{bx}{n}

\newcommand{\KL}{D_{\mathrm{KL}}}

\DeclareMathOperator*{\argmin}{arg\,min}

\usepackage{amsthm}
\usepackage{mathtools}

\DeclarePairedDelimiterX{\infdivx}[2]{[}{]}{%
  #1\;\delimsize\|\;#2%
}
\DeclarePairedDelimiterX{\infdivmi}[2]{[}{]}{%
  #1\;;\;#2%
}
\newcommand{\MI}{I\infdivmi} 
\newcommand{\KLD}{\KL\infdivx} 
\newcommand{\TV}{D_{\mathrm{TV}}}
\newcommand{\TVD}{\TV\infdivx}

\DeclarePairedDelimiter{\norm}{\lVert}{\rVert}
\DeclarePairedDelimiter{\abs}{\lvert}{\rvert}
\DeclarePairedDelimiterX{\innerProd}[2]{\langle}{\rangle}{%
    #1,#2%
}

\def\Normal{\mathcal{N}}
\def\Oh{\mathcal{O}}

\newcommand{\Nats}{\mathbb{N}}

\newcommand{\Prob}{\mathbb{P}}

\newcommand{\Ind}{\mathbf{1}}
\newcommand{\Unif}{\mathrm{Unif}}
\newcommand{\Exp}{\mathbb{E}}

\newcommand{\Laplace}{\mathcal{L}}

\DeclareMathOperator{\enc}{\mathrm{enc}}
\DeclareMathOperator{\dec}{\mathrm{dec}}

\newcommand{\XSpace}{\mathcal{X}}

\newcommand{\Exponential}{\mathrm{Exp}}

\usepackage[capitalize,noabbrev]{cleveref}

\usepackage{tikzscale}
\usepackage{pgfplots}
\pgfplotsset{compat=1.18}
\usepgfplotslibrary{fillbetween}
\usepgfplotslibrary{groupplots}
\usetikzlibrary{arrows,positioning, calc}
\usetikzlibrary {arrows.meta} 

\definecolor{tabblue}{HTML}{1f77b4}
\definecolor{taborange}{HTML}{ff7f0e}
\definecolor{tabgreen}{HTML}{2ca02c}
\definecolor{tabred}{HTML}{d62728}
\definecolor{tabpurple}{HTML}{9467bd}
\definecolor{tabbrown}{HTML}{8c564b}
\definecolor{tabpink}{HTML}{e377c2}
\definecolor{tabgray}{HTML}{7f7f7f}
\definecolor{tabolive}{HTML}{bcbd22}
\definecolor{tabcyan}{HTML}{17becf}
\definecolor{deepblue}{HTML}{0000ff}

\definecolor{colDiffc}{HTML}{1f77b4}
\definecolor{colDpDiffc}{HTML}{ff7f0e}
\definecolor{colLaplace}{HTML}{2ca02c}

\pgfplotsset{
  urp/.style={
    grid=major,
    grid style={dotted, gray!40},
    legend style={font=\small, cells={anchor=west}},
    label style={font=\small},
    tick label style={font=\footnotesize},
    width=5cm, height=3.2cm,
  },
  diffc/.style    = {colDiffc,   mark=*,         thick, mark size=1.5pt},
  dpdiffc/.style  = {colDpDiffc, mark=square*,   thick, mark size=1.5pt},
  laplace/.style  = {colLaplace, mark=triangle*,  thick, mark size=1.5pt},
  refline/.style  = {black, dashed, thick, samples=2, forget plot},
}

\icmltitlerunning{Scalable Differentially Private Data Compression via Diffusion and Stochastic Codes}

\begin{document}

\twocolumn[
  \icmltitle{Scalable Differentially Private Data Compression via \\
Diffusion and Stochastic Codes}

  \icmlsetsymbol{equal}{*}

  \begin{icmlauthorlist}
    \icmlauthor{Gergely Flamich}{imperial}
    \icmlauthor{{\"O}yk{\"u} S{\i}la G{\"u}ner}{imperial}
    \icmlauthor{Yanxiao Liu}{imperial}
    \icmlauthor{Deniz G{\"u}nd{\"u}z}{imperial}
  \end{icmlauthorlist}

  \icmlaffiliation{imperial}{Imperial College London, London, UK}

  \icmlcorrespondingauthor{Gergely Flamich}{g.flamich@imperial.ac.uk}

  \icmlkeywords{diffusion model, stochastic code, data compression, differential privacy, relative entropy coding}

  \vskip 0.3in
]



\printAffiliationsAndNotice{}  

\begin{abstract}
The ever-increasing collection of personal data has created mounting pressure to develop technologies that protect sensitive aspects of individual identity. 
Differential privacy (DP) provides a principled framework with strong formal guarantees and has already achieved practical success. 
However, releasing high-dimensional data, such as images, has remained elusive: releasing uncompressed privatized data requires significant storage. 
At the same time, no effective data compression scheme exists that can compress high-resolution data with privacy guarantees.
\par
We address this challenge with \textit{DP-DiPP}, a compression pipeline that combines stochastic codes with diffusion models. 
DP-DiPP is highly flexible: the practitioner has direct control over the compression rate-privacy-utility tradeoff.
As the theoretical backbone, we extend the Poisson private representation (PPR) of \citep{liu2024universal} to encode the outputs of privacy mechanisms.
We then combine it with DiffC, a diffusion-based lossy data compression method, to obtain a differentially private image compressor.
Our experiments on privatized image classification on CIFAR-10 demonstrate that DP-DiPP significantly outperforms the baseline, achieving a 10-30 times better compression while retaining comparable privacy guarantees and utility.
\end{abstract}
\section{Introduction}
\par
Over the past decades, the rapid growth in data collection and its use in downstream systems such as machine learning has produced undeniable benefits. 
Yet, this has also raised serious concerns about the identifiability of sensitive information and individual privacy.
Therefore, developing frameworks and algorithms that protect privacy is essential to ensure user trust. 
In this paper, we adopt \textit{differential privacy} (DP) as the privacy framework, a stronger variant of DP \citep[DP;][]{dwork2006calibrating, dwork2014algorithmic}.
DP has already found practical success, such as for collecting user data \citep{erlingsson2014rappor}, for releasing privatized tabular data, like census data \citep{bureau2021disclosure}, and for preventing generative models from memorizing and revealing their training sets \citep{liu2023differentially}.
\par
Local DP (LDP) protects individuals' data by requiring that a system's output reveal only limited information about its input.
In practice, this idea translates to randomizing the sensitive data $X$ before revealing it to an untrusted party, according to a \emph{privacy mechanism}.  
For example, if $X$ represents location data, such as GPS coordinates~\citep{andres2013geo}, one possible privacy mechanism is to add appropriately calibrated noise $\eta$ to the coordinates to prevent precise localization, that is, $Y = X + \eta$.
Of course, we still wish to use the privatized data for some purpose. Hence, a natural tradeoff arises between the privatized data's \textit{utility} for its downstream task and its privacy guarantee.
\par
However, despite its success for tabular data, scaling DP to high-dimensional data such as images, audio, or video remains challenging.
This difficulty stems from several factors.
First, directly releasing privatized high-dimensional data can incur substantial storage and communication costs.
Second, the noise required to provide meaningful privacy guarantees can severely degrade the data's utility in high-dimensional domains.
To the best of our knowledge, prior work has not systematically studied high-resolution privatized data compression, where the goal is to jointly preserve privacy, reduce storage cost, and retain downstream utility.
This is the challenge we tackle in this paper.
\par
The theoretical foundation of our work is in the emerging field of data compression with stochastic codes \citep{flamich2024data,li2024channel,flamich2026stochastic}.
In practice, stochastic codes serve as a \textit{randomized lossy compression mechanism}: for example, consider an input $X$ (modeled as a random variable) and a perturbed, lossy version $Y = X + \eta$, where $\eta$ is independent of $X$.
Then, for a given $X$, stochastic codes produce a bitstring $S$ that encodes one sample of $Y \mid X$, such that the expected length of $S$ is as small as possible.
By utilizing stochastic codes, the generated sample exactly follows the conditional distribution.
This is in sharp contrast to traditional lossy compression, where we first construct a quantizer and then attempt to determine the quantizer's error distribution. 
This flexibility of stochastic codes naturally suggests their application to compressing the output of privacy mechanisms, where the exactness of sample distribution preserves other desirable statistical properties of the original privacy mechanism, such as unbiasedness and Gaussianity.
Unfortunately, as we explain in \Cref{sec:compressing_dp}, simply encoding the outputs of privacy mechanisms with stochastic codes can lose the privacy guarantees of the original mechanism; hence, we need to design special-purpose codes that preserve the guarantees.
\par
As a solution, we use the Poisson private representation \citep[PPR;][]{liu2024universal} as a starting point for our work, a stochastic code that \textit{does} have LDP guarantees. 
Concretely, we develop an approximate version (called ``step-limited'' in this paper)  of PPR that encodes any pure $\epsilon$-LDP mechanism, ensuring that its output is at least $2\alpha\epsilon$-LDP, where $\alpha > 1$ is a tunable parameter that interpolates between the code's privacy and the expected codelength.
Unfortunately, applying PPR to complex data is challenging due to its generally very slow runtime, scaling super-exponentially with the information content.
To circumvent this issue, we combine our extended PPR code with DiffC, a variable-rate, diffusion-based data compression framework \citep{theis2022lossy}.
DiffC progressively encodes each denoising step of a Gaussian diffusion model, thereby decomposing the compression problem into smaller, manageable steps.
However, additive Gaussian noise never yields a pure $\epsilon$-LDP mechanism.
Hence, we approximate the Gaussian denoising targets at each step with moment-matched Laplace distributions, which do provide pure $\epsilon$-LDP guarantees, allowing us to apply our approximate PPR code.
Somewhat surprisingly, we empirically find that this approximation does not degrade the denoiser's performance, while significantly strengthening its formal privacy guarantee.
Thus, we call our method \textit{Differentially Private DiffC with Poisson Processes} (DP-DiPP, pronounced \textit{dippy-dip}).%
\par
Put succinctly, DP-DiPP jointly compresses and privatizes the data, rather than privatizing the data first and then compressing it losslessly.
Thus, to validate our method, we conduct experiments on privatized image compression and classification on CIFAR-10 and compare it with the naive method.
We find that DP-DiPP performs significantly better, consistently reducing bitrate by a factor of 10-30 while maintaining a comparable pure LDP guarantee and utility (as measured by classification accuracy).
\par
In summary, our main contributions are as follows: 
\begin{itemize}
\item In \Cref{sec:approximate_ppr}, we extend the Poisson private representation (PPR) of \citet{liu2024universal} to approximate sampling, and analyze its output sample quality and its DP guarantees.
Notably, we show that it yields a $2\alpha\epsilon$-LDP certificate when used to encode any pure $\epsilon$-LDP mechanism, where $\alpha \geq 1$ is a tunable parameter that trades off privacy and approximation quality.
\item In \Cref{sec:private_diffc}, we modify DiffC, a variable-rate, Gaussian diffusion-based lossy data compression method, by replacing the Gaussian targets in the denoising process with Laplace distributions, which turns the reverse process into a pure $\epsilon$-LDP mechanism.
Surprisingly, we show that this approximation does not significantly degrade the denoiser's performance.
\item In \Cref{sec:experiments}, we apply our approximate PPR code to the modified DiffC process, yielding our proposed method, \textit{DP-DiPP}, and validate its performance on a privatized image compression task on the CIFAR-10 dataset. 
Our results show that DP-DiPP consistently reduces bitrate by a factor of 10-30 while maintaining a comparable pure LDP guarantee and utility (as measured by classification accuracy).
\end{itemize}
\subsection{Notation}
\par
We denote the binary logarithm as $\log$.
For a finite set $A$, we denote the number of its elements as $\abs{A}$, for a string $s \in \{0, 1\}^*$, we denote its length by $\abs{s}$.
$\Ind[\cdot]$ denotes the indicator function.
For a real-valued function $f$, we denote its supremum norm as $\norm{f}_\infty = \sup_x \abs{f(x)}$.
For two random variables $ X$ and $ Y$ and a probability distribution $P$, we write $X \sim P$ to mean that $X$ has distribution $P$, and write $X \sim Y$ to mean that the variables are equal in distribution.
We denote uniform distributions over the interval $(a, b)$ as $\Unif(a, b)$, 
exponential distributions with rate $\lambda$ as $\Exponential(\lambda)$ and normal distributions with mean $\mu$ and variance $\sigma^2$ as $\Normal(\mu, \sigma^2)$.
Furthermore, we write $X \perp Y$ to denote the independence of $X$ and $Y$.
For two probability distributions $Q$ and $P$, we denote the relative entropy/Kullback-Leibler divergence of $Q$ from $P$ as $\KLD{Q}{P}$ (in bits) and their total variation distance as $\TVD{Q}{P}$.
For random variables $X, Y$, we denote their mutual information as $\MI{X}{Y}$ (in bits).
Finally, $\Oh(\cdot)$ is the standard big O notation.
\section{Background}
\subsection{Privacy Mechanisms and Differential Privacy}

Differential privacy (DP) provides formal certificates that quantify how much information a system reveals about its input \citep{dwork2006calibrating,dwork2014algorithmic}.
In this paper, we concern ourselves with its strict variant, local DP (LDP): we take some sensitive input $x$ and randomize it using a \textit{privacy mechanism} $M(x)$ to reduce the information leaked by $x$.
However, not any randomization will do: 
in $\varepsilon$-LDP we choose a privacy parameter $\varepsilon \geq 0$ and require the mechanism $M$ to satisfy, for any two inputs $x,x'$ and any measurable event $S$,
\[
\Prob[M(x) \in S] \leq e^\varepsilon \Prob[M(x') \in S].
\]
It can be understood as ensuring a certain level of difficulty in distinguishing whether the output of $M$ was generated from input $x$ or input $x'$, for example, when using a likelihood ratio test as the hypothesis test.

\par
Naturally, in all cases, we wish to use the data for some purpose; otherwise, we might as well replace it with pure noise to achieve perfect privacy.
Thus, we measure the usefulness of the released data through some \textit{utility function} $u$, and we are interested in maximizing the expected utility $\Exp[u(M(X))]$ under a given privacy level $\epsilon$.
Sadly, in virtually all cases of practical interest, maintaining privacy comes at the cost of utility; hence, we seek to characterize the system's \textit{utility-privacy tradeoff}~\cite{chen2020breaking}.
\subsection{Compressing Differential Privacy Mechanisms}
\label{sec:compressing_dp}
When the sensitive data is scalar-valued or low-dimensional, it is usually feasible to store or transmit it as-is, uncompressed.
However, once the size or dimensionality of the data exceeds a reasonable threshold, such as for image, audio, or video data, compression becomes necessary.
\par
The straightforward solution is to apply the mechanism $M$ first and then losslessly compress the privatized output using, e.g., \texttt{zip} for general data or \texttt{PNG} for images.
We consider lossless compression only, as using a lossy compressor could further compromise the utility of the privatized data in ways that are difficult to control.
Denoting a lossless compressor as $C$, the compressed, privatized code of the data $x$ is $C(M(x))$.
A benefit of this solution is that by the \textit{post processing} property of DP mechanisms \citep{dwork2014algorithmic}, if $M$ is $\varepsilon$-LDP, then so is $C(M(x))$.
\par
However, in many practical situations, the privatize-then-compress solution quickly reaches a limit: since privacy mechanisms randomize the data, they usually introduce new, albeit random structure.
As an extreme example, consider perfectly privatized image compression: we completely uniformly randomize the pixels of the image. 
In this case, all information is lost, but losslessly compressing each purely random pixel still requires the full bit depth of the image. 
This is sometimes referred to as the ``communication-privacy-accuracy trilemma''~\citep{chen2020breaking}. 
This counterintuitive state of affairs is essentially conclusively resolved by the emerging field of stochastic codes \citep{flamich2026stochastic}, which we discuss next.
Additionally, if the privacy mechanisms rely on continuous noise, such as Gaussian or Laplace noise, it is impossible to compress the privatized data losslessly; see~\citep{shahmiri2024communication} for a resolution of this issue for Laplace noise.
\par
The inefficiency of the compress-then-privatize approach stems from its sequential structure: the compressor only ever handles the randomized data.
In contrast, using stochastic codes allows us to \textit{jointly privatize and compress} the data, leading to significant improvements in compression rate.
Formally, given a conditional distribution $P_{M \mid X}$, a stochastic code is a triplet $(Z, \enc, \dec)$ made up of:%
\begin{itemize}%
\item A random variable $Z$ independent of the input $X$, shared between the sender/compressor and the receiver/decompressor, and hence called the \textit{common randomness}.
In practice, $Z$ can be realized by using a shared pseudo-random number generator (PRNG) and a common seed for the PRNG.
\item A function $\enc(x, z)$ mapping an input $x$ and a realization of $Z$ to a finite-length bitstring, called the \textit{encoder}.
\item A function $\dec(s, z)$ mapping a bitstring $s$ and a realization of $Z$ to the support of $M$, called the \textit{decoder}.
\end{itemize}%
Finally, to qualify as a stochastic code for $P_{M \mid X}$, the triplet must satisfy the requirement that for all $x$ we have
\begin{align}
\label{eq:stochastic_code_correctness}
\dec(\enc(x, Z), Z) \sim P_{M \mid X = x}
\end{align}
Note that in the above, the randomness arises purely from the synchronized random variable $Z$.
Intuitively, one can think of $s = \enc(x, Z)$ as encoding a single random sample from $P_{M \mid X}$.
Somewhat surprisingly, it can be shown that not only can we construct a stochastic code for virtually any conditional distribution $P_{M \mid X}$, but we can do so such that the expected codelength $\Exp[\abs{\enc(X, Z)}]$ is close to the mutual information $\MI{M}{X}$, which is optimal  \citep{li2018strong,li2024channel,flamich2024data}.
\par
As we can construct a stochastic code to compress a sample from an arbitrary conditional distribution, and privacy mechanisms may be conceived as conditional distributions, it is a natural next step to combine the two.
However, things become more involved regarding the privacy certificate: with stochastic codes, we have the Markov chain 
\begin{align}
X \to \enc(X, Z) \to \underbrace{\dec(\enc(X, Z), Z)}_{\sim M(X)}   
\end{align}
as opposed to ${X \to M(X) \to C(M(X))}$ as we had with the privatize-then-compress solution.
This observation shows that even if $M(X)$ is $\varepsilon$-LDP, post processing does not apply and $\enc(X, Z)$ is not immediately $\varepsilon$-LDP. Instead, we need to take special care in constructing $(Z, \enc, \dec)$ so that it can retain (some of) the privacy guarantees of the mechanism $M$ that it is compressing.
\par
The first works to construct such stochastic codes were
\citet{feldman2021lossless,triastcyn2021dp} and \citet{shah2022optimal}, however, they have suboptimal compression-privacy performance.
Instead, we base our work on the Poisson private representation \citep[PPR;][]{liu2024universal}, the only known stochastic code that exactly simulates arbitrary differential privacy mechanisms with optimal expected codelength. 
%
\subsection{The Poisson Private Representation}
\label{sec:ppr}
\par
As we outlined above, to define the PPR code for an $\epsilon$-LDP mechanism $M$ implementing the conditional distribution $P_{M \mid X}$, we need to specify the triplet $(Z, \enc, \dec)$ that satisfies \Cref{eq:stochastic_code_correctness}.
To begin, we assume that there is a marginal distribution $Q_M$ over the support of $M$ such that $P_{M \mid X = x} \ll Q_M$, i.e., the support of $P_{M \mid X = x}$ is contained in the support of $Q_M$ for all $x$.
Furthermore, we assume we can easily generate samples from $Q_M$ and that we can compute the density ratio (technically, the Radon-Nikodym derivative) $\frac{dP_{M \mid X}}{dQ_M}(m \mid x)$.
In our application to diffusion models, $Q_M$ will be either a Gaussian or Laplace distribution independent of $X$ while $P_{M \mid X}$ will be a Gaussian or Laplace distribution that does depend on $X$, and the above conditions will be satisfied.
Thus, we now define PPR.
%
\begin{definition}[Poisson Private Representation \citep{liu2024universal}]
\label{def:ppr}
Let $P_{M \mid X}$ be an $\epsilon$-LDP mechanism and let $Q_M$ be a probability distribution such that its support contains the support of $P_{M \mid X = x}$ for all $x$.
For convenience, let ${r(m \mid x) = \frac{dP_{M \mid X}}{dQ_M}(m \mid x)}$.
Now, let $n \in \Nats$ and $\alpha > 1$. 
Then, we define the PPR code as follows: 
%
\begin{itemize}
\item We set the common randomness as an infinite sequence of pairs of random variables ${Z \gets \{(M_k, T_k)\}_{k = 1}^\infty}$, where $M_k \sim Q_M$ and $(T_{k} - T_{k - 1}) \sim \Exponential(1)$ with the convention that $T_0 = 0$.
Such an infinite sequence is called a \textit{Poisson process}, the namesake of PPR.
\item We define the encoder in two steps.
First, using random variables $V = \{V_k\}_{k = 1}^\infty$ where $V_k \sim \Exponential(1)$ are not known to the decoder, the encoder computes
\begin{align}
\label{eq:ppr_index}
K(X, Z, V) = \argmin_{k \in \Nats} V_k \cdot T_k^\alpha \cdot r(M_k \mid X)^{-\alpha} 
\end{align}
Although the minimization is over an infinite set, it is nonetheless possible to find a minimizer in finite time. 
As $K$ is an integer, we need to map it to a bitstring to get a valid stochastic code.
To this end, we construct an arithmetic code $\zeta(k)$ from a zeta distribution, %
so that setting $\enc(X, Z, V) = \zeta(K(X, Z, V))$ and $\ell = \KLD{P_{M \mid X}}{Q_M} + \log(3.56)/\min\{(\alpha - 1)/2, 1\}$, 
\begin{align}
\label{eq:ppr_codelength}
\Exp[\abs{\enc(X, Z)}] \leq \ell + \log(\ell + 1) + 2 
\end{align}
For the precise encoding details, see~\citet{liu2024universal}.
\item Finally, the decoder simply recovers the selected index $K$ from the $\zeta$-code and then selects the sample $M_K$ that the encoder computed from $Z = \{(M_i, T_i)\}_{i = 1}^\infty$:
\begin{align}
\dec(S, Z) \gets M_{\zeta^{-1}(S)}
\end{align}
\end{itemize}
\end{definition}
The rationale for the above definition of the PPR code is threefold: first, one can show that given our selection rule in \Cref{eq:ppr_index}, we get $M_{K(x, Z, V)} \sim P_{M \mid X = x}$, and hence $\dec(\enc(x, Z), Z) \sim P_{M \mid X = x}$ as required.
Second, the expected codelength (\Cref{eq:ppr_codelength}) for any fixed $\alpha$ is close to optimal \citep{liu2024universal,li2024channel}. 
Third, Theorem 4.5 of \citet{liu2024universal} shows that when the mechanism $M(x)$ is $\epsilon$-LDP, then $\enc(x, Z)$ is $2\alpha\epsilon$-LDP, and hence, so is $\dec(\enc(X, Z), Z)$ by the post processing property. 
\par
Finally, we must deal with the runtime $N(X, Z, V)$ of PPR.
It can be shown that PPR belongs to the class of \textit{selection samplers} \citep{flamich2024some}, and thus, unfortunately, its expected runtime upon input $x$ is at least $\Exp[N(x, Z, V)] \geq \norm{r(\cdot \mid x)}_\infty$.
While this quantity can generally be quite high, it is especially problematic for our intended application to diffusion models, as in this case, the expected runtime can be infinite because the density ratios involved can be unbounded!
Thus, as our first contribution, in \Cref{sec:approximate_ppr}, we shall develop a variant of PPR that can compress privacy mechanisms with unbounded density ratios at the cost of encoding them only approximately.
However, as our \Cref{thm:approximate_ppr_properties} shows, the error of approximate PPR's output is small in total variation, and the algorithm is guaranteed to have finite runtime, a short codelength, and to retain its desirable privacy certificate.
Indeed, our algorithm will have fixed runtime, which has the further benefit that we can parallelise the most expensive computations, which further improves the runtime.
\par
Unfortunately, for the approximation guarantee to hold, the runtime has to be $\Oh(\exp(\KLD{P_{M \mid X}}{Q_M}))$, which is still prohibitively expensive for any practical privacy mechanism.
Thus, finally, we describe DiffC, a framework that enables scaling stochastic codes to high-resolution data by combining them with diffusion models.
\subsection{Data Compression with Diffusion Models}
\label{sec:diffc}
\par
A diffusion denoising probabilistic model \citep[DDPM;][]{ho2020denoising, sohl2015deep} consists of a fixed forward noising process and a learned reverse denoising process. The forward process gradually corrupts a clean sample $x_0 \sim q(x_0)$ with Gaussian noise so that, at any time-step $t \in [0, T]$, the marginal admits the closed form
\begin{align*}
x_t = \gamma_t x_0 + \sigma_t \eta_t, \quad \eta_t \sim \Normal(0, \mathrm{I}),
\end{align*}
where $\gamma_t = \sqrt{1-\sigma^2_t}$ and $\sigma^2_t$ is a monotonically increasing \emph{noise schedule} with $\sigma_0 \approx 0$ and $\sigma_T \approx 1$. We write $\mathrm{SNR}(t) = \gamma^2_t/\sigma^2_t$ for the signal-to-noise ratio at time-step $t$. The reverse process is parametrized as a Markov chain $p_\theta (x_s \mid x_t)$ trained to minimize the per-step KL divergence ${\KLD{q (x_s \mid x_t)}{ p_\theta (x_s \mid x_t)}}$ for ${s < t}$. 
Note that this objective is equivalent to minimizing ${\KLD{q (x_s \mid x_t, x_0)}{ p_\theta (x_s \mid x_t)}}$, where conditioning on $x_0$ yields a Gaussian posterior with closed form 
\begin{align}
q(x_s \mid x_t, x_0) &= \Normal(\mu_q(x_t, x_0; s, t), \sigma^2(s, t) \mathrm{I}) \label{eq:diffusion_conditional_posterior}\\
\mu_q(x_t, x_0; s, t) &= \frac{\gamma_s \sigma^2_{t\mid s}}{\sigma^2_t} x_0 + \frac{\gamma_{t\mid s} \sigma^2_s}{\sigma^2_t} x_t 
\end{align}%
where $\gamma_{t\mid s} = \frac{\gamma_t}{\gamma_s},\, \sigma^2_{t\mid s} = \sigma^2_{t} - \gamma^2_{t\mid s} \sigma^2_{s}$, and $\sigma(s, t) = \frac{\sigma_{t\mid s} \sigma_s }{\sigma_t}$.
\par
\citet{theis2022lossy} observed that the forward and reverse processes together admit a natural progressive lossy compression interpretation:
We can use a stochastic code to encode a sample $x_s$ from $q(x_s\mid x_t, x_0)$ using the model proposal $p_\theta(x_s\mid x_t)$, terminating at any intermediate timestep to trade off bitrate for distortion.
\citet{theis2022lossy} called this framework \textit{DiffC} and showed that it possesses surprisingly favorable rate-distortion performance, but did not implement it. 
Indeed, it was first implemented by \citet{vonderfecht2025lossy} using latent diffusion models, including Stable Diffusion 1.5/2.1/XL and FLUX.1-dev \citep{esser2024scaling}. 
Although the latter is a rectified flow model, stochastic sampling along its probability flow can be shown to be equivalent to DDPM sampling \citet{lai2025principles}, and hence we may use it for DiffC. 
\citet{flamich2026stochastic} recently revisited the path-based encoding in DiffC from the stochastic coding perspective.
They observed that under the assumption of a perfect denoiser, the KL divergence between joint forward and reverse paths on any common discretization $s = t_0 < t_1 < \cdots < t_n = T$ equals the KL divergence between their endpoint marginals at $s$.
This observation means that encoding the final denoising step and the entire denoising path has roughly the same compression cost using stochastic codes.
However, encoding the path is now practically feasible.
\par
Beyond porting DiffC to large pretrained models, \citet{vonderfecht2025lossy} introduced several engineering contributions that bring encoding times down to below 10 seconds. 
They perform relative entropy coding in fixed-size chunks of approximately 16 bits of KL, splitting high-KL transitions across statistically independent subsets of dimensions, so that each chunk can be feasibly compressed with a stochastic code. 
Furthermore, they use a custom CUDA kernel for the stochastic code, which draws the candidate samples directly inside the kernel rather than materializing them in GPU memory. 
Finally, they cast the choice of denoising timestep schedule as a shortest path problem over the cost matrix of pairwise per-step KL divergences. 
As an ablation, they show that hard-coding the per-encoding-step KL budget across all images, rather than transmitting it as side information, has a negligible effect on the resulting rate-distortion curves.
%
\section{Approximate PPR}
\label{sec:approximate_ppr}
\par
As we mentioned in \Cref{sec:compressing_dp}, PPR encodes an arbitrary LDP mechanism $P_{M \mid X}$ using a coding distribution $Q_M$ with as short a codelength as possible. 
However, PPR's exactness can be a drawback: its runtime scales as $\norm*{\frac{dP_{M \mid X}}{dQ_M}}_\infty$, which can be infinite even when $\KLD{P_{M \mid X}}{Q_M}$ is finite.
We address this issue by developing an approximate variant of PPR whose output is no longer an exact sample from $P_{M \mid X}$ but is still suitably close, as measured by the total variation distance.
Furthermore, this variant will retain its close-to-optimal expected codelength and its average runtime scales as $\Oh(\exp(\KLD{P_{M \mid X}}{Q_M}))$.
\par
To achieve our goal, we adopt the deceptively simple strategy of \citet{flamich2024some}: we introduce a budget $n \in \Nats$ for the maximum number of samples PPR is allowed to examine.
If the algorithm exceeds this budget, we terminate it and return the best index found so far.
\begin{definition}[Step-limited Poisson Private Representation]
\label{def:step_limited_ppr}
Let $P_{M \mid X}, Q_M, r(m \mid x)$, $Z = \{(M_i, T_i)\}_{i = 1}^\infty$, ${V = \{V_i\}_{i = 1}^\infty}$ and $\zeta$ the arithmetic code constructed from a zeta distribution as in \Cref{def:ppr}. 
In particular, recall that $P_{M \mid X}$ is an $\epsilon$-LDP mechanism whose output we wish to encode.
Now, let $n \in \Nats$ and $\alpha > 1$. 
We define the step-limited PPR algorithm as follows:
\begin{enumerate}
\item The common randomness is  $Z = \{(M_i, T_i)\}_{i = 1}^\infty$.
\item For the encoder, we define the step-limited index
\begin{align}
\label{eq:approx_ppr_index}
K_n^\alpha(x, Z, V) = \argmin_{k \in [1:n]} \left\{V_k  \left(\frac{T_k}{r(M_k \mid x)}\right)^\alpha\right\}
\end{align}
Then, we set the encoder as
\begin{align}
\enc_n^\alpha(x, Z, V) = \zeta(K_n^\alpha(x, Z, V ))
\end{align}
\item The decoder is the same as in \Cref{def:ppr}:
\begin{align}
\dec(s, Z) = M_{\zeta^{-1}(s)} 
\end{align}
\end{enumerate}
\end{definition}
We see that by definition, the runtime of step-limited PPR is finite.
But how should we choose the budget $n$?
Besides proving that our algorithm retains PPR's LDP certificate and close-to-optimal codelength, the following result shows that choosing $n = \Oh(\exp(\KLD{P_{M \mid X}}{Q_M}))$ ensures that the algorithm's output distribution will be close to that of the original mechanism.
\begin{theorem}
\label{thm:approximate_ppr_properties}
Let $P_{M \mid X}, Q_M, K_n^\alpha, V$ and let $(Z, \enc, \dec)$ be the step-limited PPR algorithm as in \Cref{def:step_limited_ppr}.
Let $n \in \Nats$ be the step-limit and $\alpha > 1$ be the PPR parameter. 
Take ${\ell_\alpha(X) = \KLD{P_{M \mid X}}{Q_M} \!+\! \frac{\log 3.56}{\min\{(\alpha - 1)/2, 1\}}}$.
Then,
\begin{enumerate}
\item \textbf{LDP guarantee:} ${\enc_n^\alpha(x, Z, V)}$ is $2\alpha\epsilon$-LDP. 
\item \textbf{Codelength:} Akin to exact PPR, we have for all $x$:
\begin{align*}
\Exp[\abs{\enc_n^\alpha(x, Z, V)}]
\leq \ell_\alpha(x) + \log(\ell_\alpha(x) + 1) + 2
\end{align*}
\item \textbf{Approximation quality:} Let $\beta > 0$ and let $Q_{M \mid X}^{(n)}$ denote the distribution of $\dec(\enc_n^\alpha(x, Z, V), Z)$.
Then, choosing $n \geq 2^{\ell_\alpha(X) / \beta}$ guarantees that
\begin{align}
\TVD{Q^{(n)}_{M \mid X = x}}{P_{M \mid X=x}} \leq \beta
\end{align}
\end{enumerate}
\end{theorem}
\begin{proof}
\textbf{LDP guarantee.}
Let $x, x'$ be two inputs.
Since the mechanism $M$ is $\epsilon$-LDP by assumption, we have 
\begin{align}
e^{-\epsilon} r(m \mid x') \leq r(m \mid x) \leq e^{\epsilon} r(m \mid x') 
\end{align}
Next, by a standard fact about the argmin of scaled exponential random variables (see Section 3.3 of \citet{li2024channel}, for example), we have
\begin{align*}
\Prob[K_n^\alpha(x, Z, V) = k \mid Z] 
&= \frac{T_k^\alpha \cdot r(M_k \mid x)^{-\alpha}}{\sum_{j = 1}^n T_j^\alpha \cdot r(M_j \mid x)^{-\alpha}} \\
&\leq \frac{T_k^\alpha \cdot (e^{-\epsilon} r(M_k \mid x'))^{-\alpha}}{\sum_{j = 1}^n T_j^\alpha\! \cdot\! (e^{\epsilon}r(M_j \mid x'))^{-\alpha}} \\
&= e^{2\alpha\epsilon}\frac{T_k^\alpha \cdot r(M_k \mid x')^{-\alpha}}{\sum_{j = 1}^n T_j^\alpha \cdot r(M_j \mid x')^{-\alpha}} \\
&= e^{2\alpha\epsilon} \Prob[K_n^\alpha(x', Z, V) = k \mid Z] 
\end{align*}
as desired.
\par
\textbf{Codelength.}
By its definition $K_n^\alpha(x, z, v) \leq K(x, z, v)$, and hence we have that
\begin{align}
\Exp[\log K_n^\alpha(x, Z, V)] \leq \ell_\alpha(x)
\end{align}
Thus, using our arithmetic code $\zeta$ to encode $K_n^\alpha(x, Z, V)$, we get the desired result.
\par
\textbf{Approximation quality.}
We use a technique similar to that of \citet{flamich2024some}:
Let $K(X, Z, V)$ be the index of exact PPR defined in \Cref{eq:ppr_index}.
First, note that (e.g., using the technique from Section 4.3.2 of \citet{flamich2024data}):
\begin{align*}
\TVD{Q^{(n)}_{M \mid X = x}}{P_{M \mid X = x}} 
&\leq \Prob[K(X, Z, V) > k]
\end{align*}
By Theorem 4.3 of \citet{liu2024universal} we have for $\ell(X) = \KLD{P_{M \mid X}}{Q_M}$ that
\begin{align}
\Exp[\log K_n^\alpha(X, Z, V)] \leq \ell_\alpha(X) 
\end{align}
Since $\log K(X, Z, V) \geq 0$, we may apply Markov's inequality to obtain
\begin{align}
\Prob[K(X, Z, V) > n] \leq \frac{\ell_\alpha(X)}{\log(n + 1)}
\end{align}
Therefore, we see that
\begin{align}
\TVD{Q^{(n)}_{M \mid X = x}}{P_{M \mid X = x}} \leq \frac{\ell_\alpha(X)}{\log(n + 1)}
\end{align}
Thus, setting $n \geq 2^{\frac{\ell_\alpha(X)}{\beta}}$ yields the desired approximation guarantee.
\end{proof}
\section{Privatized Data Compression with Diffusion}
\label{sec:private_diffc}
\subsection{The Denoising Process as a Privacy Mechanism}
\label{sec:diff_priv}
Our most important high-level observation is that, \textit{the DDPM denoising process conditioned on the original image $x_0$ is a valid privacy mechanism}. 
Concretely, at each denoising timestep, we may view a sample from the conditional posterior given in \Cref{eq:diffusion_conditional_posterior} as an additive Gaussian privacy mechanism for $x_0$:
\begin{align}
\label{eq:privacy_mechanism}
M(x_0) \!=\! \mu_q(x_t, x_0; s, t) + \sigma(s, t)\, \eta,\quad 
\eta \sim \Normal(0, \mathrm{I})
\end{align}
We defer the discussion on the privacy guarantee of this mechanism to \Cref{sec:appA}. 
\par
Unfortunately, the Gaussian mechanism does not yield a \textit{pure} $\epsilon$-LDP certificate, and hence our \Cref{thm:approximate_ppr_properties} does not apply if we apply step-limited PPR to it.
To circumvent this issue, we moment-match Laplace distributions to the Gaussian targets and proposals at each time step and use those for denoising and compression instead. 
Thus, the privacy mechanism becomes
\begin{align}
\label{eq:privacy_mechanism_laplace}
M(x_0) = \mu_q(x_t, x_0; s, t) + b(s, t)\, \eta,\quad \eta \sim \Laplace(0, \mathrm{I})
\end{align}
where $b(s, t) = \sigma(s, t) / \sqrt{2}$. 
For a discussion of high-dimensional Laplace distribution on differential privacy, see~\citep{andres2013geo}. 

In what follows, we aim to provide per-pixel guarantees for image data.
Thus, $x^{(i)}_0$ will denote a pixel on an image, consisting of $C$ channels (e.g., $C=3$ for RGB), with $x^{(i,c)}_0$ denoting channel $c$.
We start with per-subpixel guarantees, which we also refer to as ``per-channel,'' from which we recover the per-pixel guarantees.
Following the steps in \citet[Theorem 3.6]{dwork2014algorithmic} applied to a single channel, the per-channel mechanism satisfies $\epsilon_\text{ch}$-DP where 
\begin{align}
\epsilon_\text{ch} = \frac{\gamma_s \sigma^2_{t\mid s} }{\sigma^2_t b(s,t)}\cdot \sup\abs{ x^{(i,c)}_0 - x'^{\,(i,c)}_0}\label{eq:subpixel_guarantee}
\end{align}
and the supremum is taken over all possible images. This admits the nicer form (see \Cref{sec:appA} for details):
\begin{align*}
\epsilon_\text{ch} = \sqrt{2(\text{SNR}(s)-\text{SNR}(t))} \sup{\vert x^{(i,c)}_0 - x'^{\,(i,c)}_0}\vert
\end{align*}
By basic composition over the $C$ channels of a pixel, the per-pixel $\epsilon_{t \rightarrow s}$-DP guarantee per step is
\begin{align*}
\epsilon_{t\rightarrow s} = C\sqrt{2(\text{SNR}(s)-\text{SNR}(t))} \sup{\vert x^{(i,c)}_0 - x'^{\,(i,c)}_0}\vert
\end{align*}
\subsection{Compressing the Privacy Mechanism}
In DiffC~\citep{theis2022lossy}, the output of each denoising step is not released directly, but compressed via a stochastic code; in our case, this is where we will apply the PPR code. 
Under PPR, at each step we compress the privacy mechanism $M_{t\rightarrow s}(x_0)$ and the compressed output inherits a privacy guarantee that depends on PPR's parameter $\alpha$ and the privacy parameters of the mechanism being compressed.
Note that in the works of \citet{theis2022lossy} and \citet{vonderfecht2025lossy}, the code they used is equivalent to our PPR code with the parameter $\alpha = \infty$, which achieves maximal compression efficiency but provides no privacy guarantee.
\par
However, as we explained in \Cref{sec:ppr}, exact PPR (that is, the budget is $n = \infty$) requires an upper bound on the density ratio between the target and the proposal distributions. In the Gaussian case of \Cref{eq:privacy_mechanism}, both distributions are Gaussians with the same variance
\begin{align}
\label{eq:denoising_distributions}
q = \Normal(\mu_q, \sigma^2 \mathrm{I}),\, p = \Normal(\mu_p, \sigma^2 \mathrm{I}),
\end{align}
whose density ratio is easily seen to be unbounded. 
Fortunately, with our moment-matched Laplace distributions (\Cref{eq:privacy_mechanism_laplace}) this is no longer an issue: the density ratio of two Laplace distributions with the same scale is always bounded. 
However, despite the density ratio no longer being an issue, step-limited PPR still provides a significant computational advantage compared to its exact variant.
Since we choose the step-budget $n$ ahead of time, we can draw the $n$ proposal samples and evaluate their density ratios in parallel on a GPU using a custom CUDA kernel (see \Cref{sec:custom_cuda_kernel}).
\par
Thus, we use step-limited PPR (\Cref{def:step_limited_ppr}). 
Combining it with the per-step, per-pixel guarantee derived in \Cref{sec:diff_priv}, the PPR code satisfies a $2\alpha\epsilon_{t\rightarrow s}$-LDP per-pixel guarantee at each step. 
By sequential composition, the total per-pixel privacy certificate across $n$ encoding steps is:
\begin{align*}
\epsilon_\text{total} = 2 \alpha \sum_{i=1}^n \epsilon_{t_i\rightarrow t_{i-1}}
\end{align*}

We show experimentally that, for compression to the same final timestep, replacing the Gaussian channel with the moment-matched Laplace channel and $\alpha = \infty$ with $\alpha = 2$ in step-limited PPR has a negligible effect on the downstream task, while increasing the bitrate by a factor of $\approx 2$; see the bottom left panel of \Cref{fig:comparison_plot}.
The factor $2$ increase can be in large part explained by comparing the relative entropies between the original Gaussians involved and their moment-matched Laplace counterparts, see \Cref{sec:factor_two_increase_explanation}.

\section{Experiments}
\label{sec:experiments}
\begin{figure*}[t]
\centering
\includegraphics{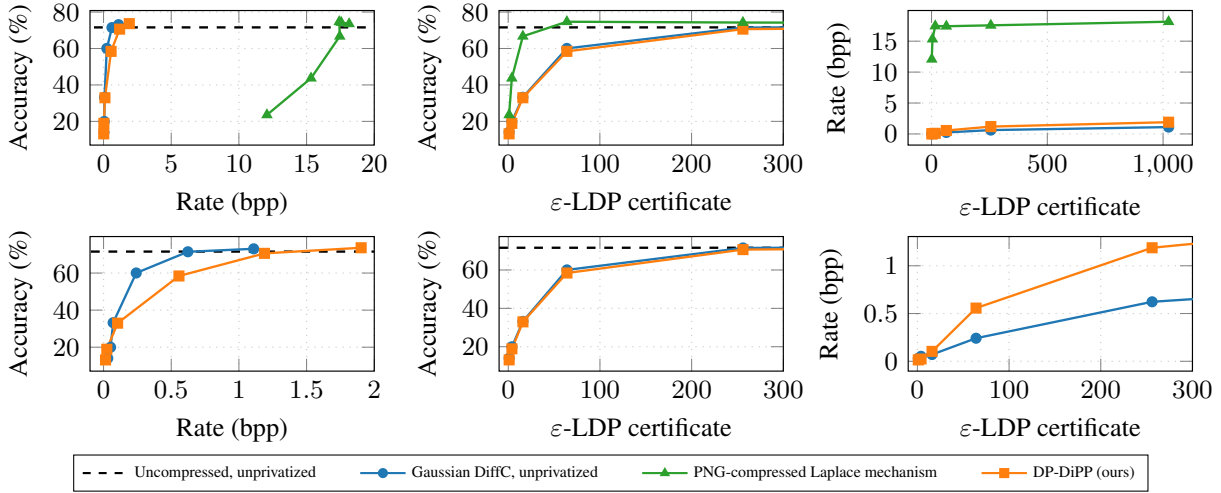}
\ref*{legend:methods}
\caption{%
Experimental results on training an image classifier from a dataset of privatized and compressed CIFAR-10 images.
Note that Gaussian DiffC \textbf{does not} have an LDP certificate and is included only for illustration.
The left and middle columns also include the classification accuracy from the original, unprivatized data for comparison.
In each case, the markers are connected by lines \textbf{in order of increasing $\epsilon$}, where $\epsilon$ is the method's LDP guarantee.
\textbf{Left column:} The rate-utility tradeoff, as measured by bits per pixel (bpp) and test classification accuracy, respectively.
\textbf{Middle column:} Privacy-utility tradeoff, where the method's LDP certificate measures privacy.
\textbf{Right column:} Rate-privacy tradeoff.
\textbf{Top row:} comparisons among 1) the privatize-then-compress method using 2) Gaussian DiffC and 3) DP-DiPP.
\textbf{Bottom row:} Same comparisons as top row with the privatize-then-compress method omitted.
}
\label{fig:comparison_plot}
\end{figure*}
We compare the privacy-utility-compression tradeoff of DP-DiPP against the privatize-then-compress baseline and DiffC \citep{theis2022lossy} for training an image classifier on privatized and compressed CIFAR-10 images, taking final classification accuracy as our utility function.

Concretely, we add calibrated Laplace noise to the image and compress the result with the \texttt{PNG} lossless image compression codec. 
Both methods satisfy per-pixel pure $\epsilon$-DP, so their utility and compression rate can be compared at the same privacy level. We also compare against the original DiffC to understand how replacing the Gaussian channel with the Laplace channel as well as using $\alpha = \infty$ versus $\alpha = 2$ for PPR affects utility and coding cost.
\subsection{Privatization}
\par
For privatization, our encoding pipeline modifies DiffC in two ways: we sample from the moment-matched Laplace channel ($b = \sigma/\sqrt{2}$), and we compress this channel with PPR using $\alpha = 2$.
At each encoding step, our CUDA kernel samples from $z_i \overset{\text{iid}}{\sim}\Laplace(0, 1)$ and computes the logarithm of density ratios evaluated at those samples as  $\sum_i(\abs{z_i} - \abs{z_i - \mu_i})$ (except the constant term), differently from the implementation in \citet{vonderfecht2025lossy} for PFR, which is $\sum_iz_i\mu_i$ (here $\mu_i = (\mu_q-\mu_p)/b$). 
The encoder then uses these precomputed density ratios to run step-limited PPR (\Cref{def:step_limited_ppr}). 
To decode, the receiver, sharing the same random seed, retrieves the selected sample and reconstructs $x_s = b \cdot z + \mu_p$ using the diffusion model's noise prediction
\begin{align*}
    \mu_p(x_t, \eta_\theta; s, t) = \frac{1}{\gamma_{t \mid s}}\left(x_t - \frac{\sigma_{t \mid s}^2}{\sigma_t}\eta_\theta(x_t, t)\right)
\end{align*}
Each encoded image produces two outputs. The \emph{noisy reconstruction} is the raw output from the final encoding step; it lies at an intermediate diffusion timestep and contains visible noise. The \emph{denoised reconstruction} is obtained by running DDPM from this final timestep to $t = 0$.

The encoding schedule is chosen by solving a resource-constrained shortest-path problem over the timesteps, extending the method used in \citet{vonderfecht2025lossy}, which does not have a privacy constraint. The aim is to find the schedule that minimizes the total coding cost subject to the privacy budget $\epsilon_\text{total}$ for each final timestep $T_\text{final}$.
Then, setting $\Delta_{\text{ch}} = {\sup\abs{ x^{(i,c)}_0 - x'^{\,(i,c)}_0}} = 2$, encoding each denoising step consumes 
\begin{align*}
\epsilon_\text{PPR}{(t\rightarrow s)} = 6\alpha\sqrt{2 (\text{SNR}(s) - \text{SNR}(t))} \Delta_{\text{ch}}
\end{align*}
units of per-pixel privacy budget.

For the Laplace noising baseline, we add independent $\Laplace(0, \Delta'_\text{ch}/(\epsilon/3))$ noise to each channel where $\Delta'_\text{ch} = 255$, clamp to $[0, 255]$, and compress losslessly with \texttt{PNG}.
\subsection{Classification}
The pretrained DDPM model \citep{ho2020denoising} used as the backbone of our compression schemes has been trained on the CIFAR-10 training set.
Hence, to prevent data leakage, we use the CIFAR-10 test set, split evenly into 5k training and 5k evaluation images to train and evaluate our classifier.

To measure how privacy affects downstream utility, we follow the protocol of \citet{cao2026ldpslicing} and train a ResNet-56 classifier directly on the privatized images. The classifier never sees clean data; it is trained in the noisy domain using the provided image labels, using the same hyperparameters as \citet{cao2026ldpslicing} to enable direct comparison.
\subsection{The DiffC Baseline}
\par
We compare DP-DiPP to the original Gaussian DiffC \citep{vonderfecht2025lossy}, which is equivalent to our step-limited PPR algorithm with $\alpha = \infty$.
Note that Gaussian noise lacks a pure LDP guarantee, and so does using PPR with $\alpha = \infty$;
we include it as an ablation to show that DP-DiPP's performance does not degrade significantly when using moment-matched Laplace distributions instead of the original Gaussians.
To ensure a fair comparison between DiffC and DP-DiPP, we first solve the resource-constrained shortest path problem for the DP-DiPP mechanism at each privacy level $ \epsilon$. 
After we decide on a final timestep $T_\text{final}$, we find the shortest path for DiffC that leads to $T_\text{final}$. 

%
\subsection{Empirical Findings}
\par
Our experimental results are shown in \Cref{fig:comparison_plot}.
Each point on the curve represents a final compression timestep, after which each method uses the DDPM denoiser to complete denoising the images.
For our experiments, we tuned each method to attain a per-pixel $\epsilon$-LDP certificate (or equivalent in the case of DiffC) for $\epsilon \in \{1, 4, 16, 64, 256\}$.
For the step-limited PPR code in DP-DiPP, we set $\alpha = 2$.
\par
First and foremost, our findings validate our method: DP-DiPP's joint privatization and compression significantly outperform the privatize-then-compress baseline, as shown in the top row of \Cref{fig:comparison_plot}.
In particular, DP-DiPP achieves the same classification accuracy (top left panel) and the same LDP guarantee (top right panel) using less than 3-11\% of the bits that the baseline method needs, yielding a 10-30$\times$ improvement in the rate!
\par
Second, in the bottom row of \Cref{fig:comparison_plot}, comparison of DiffC and DP-DiPP shows that using moment-matched Laplace distributions instead of Gaussians and a strict privacy parameter of $\alpha=2$ introduces a roughly 2-times penalty in the rate (bottom right), while the utility is largely unaffected.
\section{Discussion and Future Directions}
We introduced an approximate version of PPR and showed that it retains an LDP guarantee and has good approximation properties.
Then, we applied it to diffusion-based compression, and showed that using Laplace instead of Gaussian denoising distributions can significantly improve the formal privacy guarantee while retaining empirical performance.
We applied our method to image compression on CIFAR-10 and showed that our joint compression method significantly outperforms the privatize-then-compress baseline.
\par
We have only dealt with pure $\epsilon$-LDP mechanisms.
Thus, a future theoretical direction could be to consider whether it is possible to construct approximate versions of PPR from mechanisms that are $(\epsilon,\delta)$-LDP instead.
\section*{Impact Statement}
This paper presents work aimed at advancing the field of Machine Learning. There are many potential societal consequences of our work, none of which we feel must be specifically highlighted here.

\section*{Acknowledgements}
The authors acknowledge financial support from Imperial College London through an Imperial College Research Fellowship grant awarded to GF, and from the UKRI for projects AI-R (ERC Consolidator Grant, EP/X030806/1) and INFORMED-AI (EP/Y028732/1).

\bibliography{references}
\bibliographystyle{icml2026}

\newpage
\appendix
\crefalias{section}{appendix}
\onecolumn

\section{Notes on Privacy Guarantee of Gaussian Diffusion}
\label{sec:appA}
We can investigate the privacy guarantee for the Gaussian diffusion process, continuing from \Cref{eq:privacy_mechanism}. 
The $\ell_2$ sensitivity of this mechanism is 
\begin{align*}
\Delta_{t\rightarrow s} 
&= \frac{\gamma_s \sigma^2_{t\mid s}}{\sigma^2_t} \sup_{x_0, x'_0 \in \XSpace} \Vert x_0 - x'_0 \Vert_2\\
&= \frac{\gamma_s \sigma^2_{t\mid s}}{\sigma^2_t} D_\XSpace.\\
\end{align*}
and the $\Delta/\sigma$ ratio simplifies to 
\begin{align*}
\frac{\Delta_{t\rightarrow s}}{\sigma(s, t)}
&= \frac{\gamma_s \sigma^2_{t\mid s}/\sigma^2_t}{\sigma_{t\mid s} \sigma_s / \sigma_t} D_\XSpace\\
&= \frac{\gamma_s \sigma_{t\mid s}}{\sigma_t \sigma_s} D_\XSpace\\
&=\sqrt{\left(\frac{\gamma_s}{\sigma_s}\right)^2 - \left(\frac{\gamma_t}{\sigma_t}\right)^2} D_\XSpace\\
&=\sqrt{\text{SNR}(s) - \text{SNR}(t)} D_\XSpace.
\end{align*}

The $(\epsilon, \delta)$-differential privacy of this mechanism at each time step can be computed tightly using the result of \citet[Theorem 8]{balle2018improving}. The privacy guarantee of the full multi-step mechanism, which is the composition of per-step mechanisms, can then be obtained via an adaptive composition theorem such as \citet[Theorem 3.5]{pmlr-v37-kairouz15}.
\section{Custom CUDA Kernel for PPR}
\label{sec:custom_cuda_kernel}

Different from the implementation released by the authors of \citet{vonderfecht2025lossy}, to sample from a moment matched laplace distribution at each timestep, we first generate laplace samples with mean $0$ and scale $1$. We draw these samples using inverse transform sampling:

\begin{minted}[
  bgcolor=gray!8,
  linenos,
  breaklines,
  fontsize=\small
]{cpp}
__device__ float curand_laplace_standard(curandState* state) {
    float u = curand_uniform(state);
    float half_minus = u - 0.5f;
    float sign_val = (half_minus >= 0.0f) ? 1.0f : -1.0f;
    float abs_val = fabsf(half_minus);
    float inner = fmaxf(1.0f - 2.0f * abs_val, 1e-10f);
    return -sign_val * logf(inner);  // Laplace(0, 1)
}
\end{minted}

Also differently from \citet{vonderfecht2025lossy}, the log density ratio ($\log w(z)$) between target $\mathcal{L}(\mu, 1)$ (where $\mu = \frac{\mu_q - \mu_p}{b}$) and proposal $\mathcal{L}(0, 1)$ analytically becomes
\begin{align*}
\log w(z) = \abs{z} - \abs{z-\mu}
\end{align*}
thus we generate these ratios as:
\begin{minted}[
  bgcolor=gray!8,
  linenos,
  breaklines,
  fontsize=\small
]{cpp}
    float log_w_value = 0.0f;
    for (int i = 0; i < dim; i++) {
        float z = curand_laplace_standard(&state);
        log_w_value += fabsf(z) - fabsf(z - mu_q[i]);
    }
\end{minted}

We decide on the sample returned by step limited PPR using:
\begin{minted}[
  bgcolor=gray!8,
  linenos,
  breaklines,
  fontsize=\small
]{python}
    t = cp.random.exponential(scale=1.0, size=K)
    v = cp.random.exponential(scale=1.0, size=K)
    
    s = alpha * (log_cumsum_t - log_w) + log_v

    winning_seed = cp.argmin(s).item()
    sample = generate_sample(dim, shared_seed, winning_seed)
\end{minted}

We note that we can convert samples from $Z_i \overset{\text{iid}}{\sim} \mathcal{L}(\frac{\mu_q - \mu_p}{b}, 1)$ to $X_i \overset{\text{iid}}{\sim}\mathcal{L}(\mu_q, b)$ by
\begin{align*}
X_i = Z_i \cdot b + \mu_p,
\end{align*}
where both $\mu_p$ and $b$ are available to both the encoder and decoder. This is implemented in the code as
\begin{minted}[
  bgcolor=gray!8,
  linenos,
  breaklines,
  fontsize=\small
]{python}
            mu = ((qmu[i] - pmu[i]) / b).flatten().detach().cpu().numpy()
            sample, cs, dkl = ch_sim.encode(mu, manual_dkl=manual_dkl[si], seed=si)
            seeds_all[i].append(cs); dkl_all[i].append(dkl)
            xl[i] = torch.tensor(sample).reshape(xl[i].shape).to(dev).to(dt) * b + pmu[i]
\end{minted}

\section{Explanation of the Factor Two Increase in Bitrate after Moment Matching}
\label{sec:factor_two_increase_explanation}
We illustrate the phenomenon for one-dimensional distributions.
As the distributions involved in practice are fully factorised dimensionwise, the conclusion carries over by the chain rule of relative entropies.
\par
Thus, consider now
\begin{gather*}
q = \Normal(\mu_q, \sigma^2),\quad p = \Normal(\mu_p, \sigma^2), \\
q' = \Laplace(\mu_q, \sigma / \sqrt{2}),\quad p' = \Laplace(\mu_p, \sigma / \sqrt{2})
\end{gather*}
where $q, p$ represent one dimension of the conditional and unconditional denoising distributions from \Cref{eq:denoising_distributions}, and $q', p'$ are their moment-matched counterparts.
For convenience, let $\delta = \frac{\abs{\mu_p - \mu_q}}{\sigma}$ be the standardised absolute mean deviation.
Then, using the well-known formulae for the relative entropy of two Gaussian and two Laplace distributions, we get
\begin{align}
\KLD{q}{p} \cdot \ln 2 &= \frac{\delta^2}{2} \label{eq:gauss_kl}\\
\KLD{q'}{p'} \cdot \ln 2 &= \sqrt{2}\delta - 1 + e^{-\sqrt{2}\delta}  \\
&= \sqrt{2}\delta - 1 + 1 - \sqrt{2}\delta + \delta^2 + \Oh(\delta^3) \tag{Taylor expansion around $0$} \\
&= \delta^2 + \Oh(\delta^3) \label{eq:taylor_simple_lap_kl}
\end{align}
Comparing \Cref{eq:gauss_kl,eq:taylor_simple_lap_kl}, we see that for small $\delta$, we indeed have $\KLD{q'}{p'} \approx 2  \cdot \KLD{q}{p}$, as we observe in practice.
\begin{figure}[t]
\centering
\includegraphics[scale=.8]{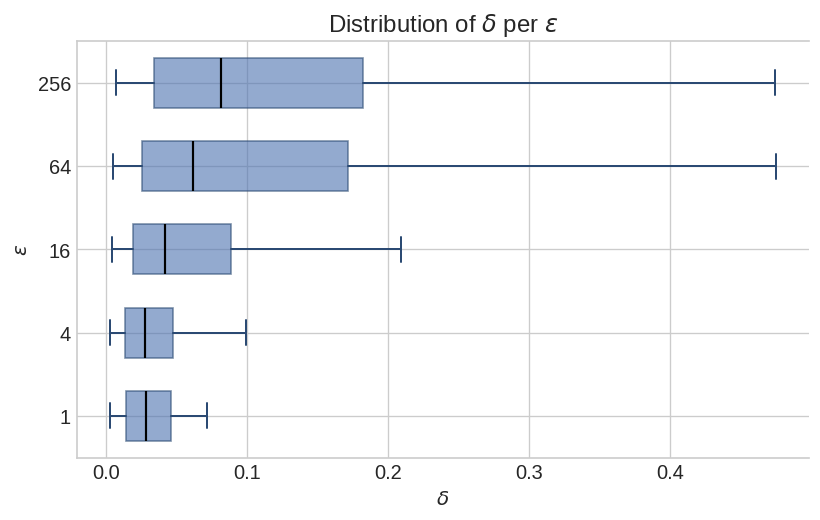}
\caption{Per $\epsilon$ distribution of the standardised absolute mean deviation $\delta$. Each box pools $\delta$ over all time steps of the PPR schedule, all 128 test images, and all $32 \times 32 \times 3$ latent dimensions. Box = interquartile range, line = median, whiskers = 5th/95th percentiles.}
\label{fig:delta_histogram} 
\end{figure}
We corroborate this empirically, by studying the distribution of $\delta$ across the dimensions of our diffusion models for different $\epsilon$-LDP settings in \Cref{fig:delta_histogram}.
We see that for all settings, the distribution of $\delta$ is heavily skewed towards $0$, with barely any values being close even to $1$; thus our Taylor approximation is roughly valid.

\end{document}